\documentclass[letterpaper, 10 pt, conference]{ieeeconf}  % Comment this line out
\IEEEoverridecommandlockouts                              % This command is only
\author{Bruce Lee and Andrew Lamperski% <-this % stops a space
\thanks{This work was supported in part by NSF CMMI-1727096}% <-this % stops a space
      \thanks{The authors are with the department of Electrical and
        Computer Engineering, University of Minnesota, Minneapolis,
        MN 55455, USA 
        {\tt\small leex8370@umn.edu, alampers@umn.edu}}
}

\title{\LARGE \bf
  Non-asymptotic
  Closed-Loop System Identification using Autoregressive Processes and Hankel Model
  Reduction}

\usepackage{xcolor}
\usepackage{mathtools}
\usepackage{amssymb}
\usepackage[overload]{empheq}
\let\proof\relax
\let\endproof\relax
\usepackage{amsthm}
\usepackage{ifthen}
\newtheorem{theorem}{\textbf{Theorem}}

\newtheorem{lemma}{\textbf{Lemma}}

\newtheorem{remark}{\textbf{Remark}}

\usepackage{tikz}
\usetikzlibrary{shapes,matrix,arrows,calc,positioning,fit}
\tikzset{
dashedx/.style={}
}

\usepackage{pgf}
\usepackage{pgfplots}
\pgfplotsset{width=9.1cm,height=5.8cm,compat=newest}
\pgfplotscreateplotcyclelist{plain black}{%
	every mark/.append style={fill=gray,scale=0.6},mark=\star\\%
}

\usepackage{algorithm}
\usepackage{algpseudocode}
\usepackage{graphicx}
\usepackage[framemethod=tikz]{mdframed}
\definecolor{mycolor}{rgb}{1, 0, 0}

\newmdenv[innerlinewidth=0.5pt, roundcorner=4pt,linecolor=mycolor,innerleftmargin=6pt,
innerrightmargin=6pt,innertopmargin=6pt,innerbottommargin=6pt]{mybox}

\newcommand{\argmin}{\operatornamewithlimits{argmin}}

\renewcommand{\H}{\mathbf{H}}
\newcommand{\E}{\mathbb{E}}
\newcommand{\A}{\mathbf{A}}
\newcommand{\B}{\mathbf{B}}
\newcommand{\C}{\mathbf{C}}
\newcommand{\D}{\mathbf{D}}

\newcommand{\K}{\mathbf{K}}
\renewcommand{\P}{\mathbb{P}}
\newcommand{\y}{\mathbf{y}}
\newcommand{\yStar}{\mathbf{y}_{t|t-p:t}^\star}
\newcommand{\yKF}{\mathbf{y}_{t|-\infty:t}^\star}

\renewcommand{\v}{\mathbf{v}}

\newcommand{\x}{\mathbf{x}}

\renewcommand{\u}{\mathbf{u}}
\newcommand{\e}{\mathbf{e}}

\newcommand{\s}{\mathbf{s}}

\newcommand{\z}{\mathbf{z}}
\newcommand{\Z}{\mathbf{Z}}
\renewcommand{\d}{\mathbf{d}}
\newcommand{\Q}{\mathbf{Q}}
\newcommand{\G}{\mathbf{G}}
\newcommand{\Pow}{\mathcal{P}}

\newcommand{\TKF}{\mathbf{Head}}
\newcommand{\Tr}{\mathrm{Tr}}
\newcommand{\Hinf}{\mathcal{H}_{\infty}}
\newcommand{\Tail}{\mathbf{Tail}}

\newcommand{\R}{\mathbb{R}}
\newcommand{\N}{\mathbf{N}}

\begin{document}
\newcounter{num}
\setcounter{num}{1}
%1 arxiv
%0 ACC

\maketitle
\begin{abstract}
  
  %%%%% Motivate the problem %%%%%
  % Non-asymptotic bounds are good, because in principle, they give a
  % bound on the amount of data needed for adequate ID
  %
  % Several system ID results give non-asymptotic bounds for open-loop ID
  %
  % However, many times data is collected in closed loop.
  % 
  % Applying open-loop ID to closed-loop data can result in biases
  %
  % ONe method from subspace ID used to eliminate these biases
  % involves first fitting long-horizon autoregressive model, and then
  % performing a model reduction.
  % 
  % The asymptotic behavior of these algorithms is well characterized
  % but their non-asymptotic behavior is not.
  %
  % We characterize it (at least for one particular variant)
  %
  % THen we get specific about result in expectation and in high probability
  %%%%%% Describe our solution %%%%
  % We get non-asymptotic bounds on closed-loop ID
    One of the primary challenges of system identification is 
    determining how much data is necessary to adequately fit a model.
    Non-asymptotic characterizations of the performance of system
    identification methods provide this knowledge. Such 
    characterizations are available for several algorithms
    performing open-loop identification. Often times, however,
    data is collected in closed-loop. Application of open-loop
    identification methods to closed-loop data can result
    in biased estimates. One method used by subspace
    identification techniques to eliminate these biases involves
    first fitting a long-horizon autoregressive model, then
    performing model reduction. The asymptotic behavior of
    such algorithms is well characterized, but the non-asymptotic
    behavior is not. This work provides a non-asymptotic
    characterization of one particular variant of these
    algorithms. More specifically, we provide non-asymptotic 
    upper bounds on the generalization error 
    of the produced model, as well as high probability
    bounds on the difference between the produced
    model and the finite horizon Kalman Filter.
    % This work provides a non-asymptotic upper bound
    % on the expected squared prediction error for
    % the model produced by
    % a system identification algorithm consisting of an autoregressive modeling step followed
    % by balanced model reduction when identification
    % is performed in closed-loop.
    % Additionally, high probability bounds
    % are supplied for the $\Hinf$ norm
    % of the error system from the model output
    % by the identification algorithm to the finite
    % horizon Kalman Filter.
\end{abstract}

\section{Introduction}

% What is the motivation?
One of the first steps in the control design process is to obtain a model
for the system of interest. In cases where knowledge of the system is nonexistent
or incomplete, models must be identified from input/output data.
This process can be viewed as a learning problem in which models are
optimized in order to give the best fit for the 
data \cite{ljung1999system}.
%
% The main goal of system identification is to find a model that
% explains current data, and predicts new data well.
%
The quality of the model can be assessed via 1) generalization
error, which measures how well the model fits unseen data, and 2)
model error, which measures how far the identified model is from the
``true'' model. (In many cases, analysis of model error is an
idealization, since the real system falls out side the class of models
studied.)

% two closely related
% ways. Approaches 
% %
% A closely related problem is bounding the uncertainty about the model
% that has been identified.
%

% What are the challenges?
System identification can be viewed as a learning
problem, but correlations in the data lead to several challenges. 
Typical
machine learning problems assume that the data are
independent \cite{mohri2018foundations}. Using independence, learning
theory provides \emph{non-asymptotic} bounds on the
generalization error obtained from finite amounts of data. 
In contrast, the data from system identification are correlated due to 1) internal system
dynamics, 2) temporal correlations in the inputs, and 3) feedback from
the outputs to the inputs.
The result is that most traditional analyses of system identification
methods focus on \emph{asymptotic} bounds, which can only guarantee
low generalization error in the limit of infinite data
\cite{ljung1999system}. There has, however, been recent efforts
to provide non-asymptotic bounds on the performance
of system identification methods.
%An overview of some of these efforts is provided below
%

% Overview of what we solve
%This paper gives non-asymptotic generalization bounds for a
%closed-loop subspace identification method. Furthermore, in the case
%that the data is generated by a stable linear system, we can bound the
%error between our identified system and the true system with respect
%to the $\Hinf$ norm.
% The sentences above set this up.
%Below, we overview available literature on non-asymptotic system
%identification and closed-loop subspace identification.

Most work on non-asymptotic system identification focuses on open-loop
problems. Early works give 
non-asymptotic 
analyses for the identification of
transfer functions \cite{goldenshluger1998nonparametric} and
autoregressive models for
systems with no measured inputs \cite{goldenshluger2001nonasymptotic}. 
Recently, several works have provided non-asymptotic analyses of various
open-loop system identification problems for 
stable linear time invariant
systems. The work in
\cite{tu2017non}
bounds the error 
in fitting a finite impulse response with %independent, identically distributed inputs.  
inputs chosen optimally for identification. The case where the
the state is measured directly and the inputs are independent and identically distributed (iid)
is studied in \cite{simchowitz2018learning}.
The work in \cite{oymak2018non},\cite{sarkar2019finite} bounds the error in identifying a finite impulse
response and obtaining a realization from data generated with iid inputs.

The work in \cite{hazan2017learning,hazan2018spectral} provides a
non-asymptotic method for output error identification of linear
models. Unlike the works mentioned above, the data could be collected
in closed-loop. However, these works utilize the non-probabilistic
framework of online optimization
\cite{hazan2016introduction,cesa2006prediction}, and are not
directly comparable to the work on generalization
bounds. Additionally, the models identified in these works are
restricted to stable systems.

% Characterizing
% this uncertainty is of great interest, as it allows the engineer to understand
% the quality of the derived model. Furthermore, if the model is to be used for robust
% control, it is essential to have a characterization
% of the model uncertainty. 
For control design, quantifying the error between the identified model
and the ``true'' model is useful. 
An overview of methods for control design from identified models is
provided in \cite{gevers2005identification}.
Recent approaches to robust control synthesis
that take the uncertainty of identified models into account
are analyzed in \cite{tu2017non} and \cite{dean2017sample}.
%Each of these papers describes a process in which
%a system model is identified in open-loop along with high probability uncertainty bounds.
%A controller is then synthesized using the estimated model 
%and uncertainty bounds.

% The importance of uncertainty bounds for identified models has driven a vast
% amount of study. Results regarding the asymptotic convergence properties
% of least squares for estimating transfer functions are provided in 
% \cite{ljung1985asymptotic} and \cite{ljung1992asymptotic}. . 

As discussed above, the recent works on non-asymptotic identification
have focused almost exclusively on open-loop identification methods. 
% While the recent efforts have served to provide characterizations of model 
% uncertainty for systems identified from open-loop experiments,
% the analyses have not explored the realm of closed-loop identification. 
%Such shortcomings are highly limiting, as
However, for many systems, 
the plant is impossible to isolate from
its controller or is unstable in open-loop. Furthermore, identification is most
successful when performed in circumstances that closely match the desired
application, which often includes a feedback 
controller \cite{gevers2005identification}. This 
drives the study of methods that are effective with closed-loop data.

The task of developing identification methods
that work on closed-loop data is nontrivial, as the correlation between past output noise and future inputs
produces a bias in model estimates for many identification methods.
This is particularly troublesome for subspace approaches \cite{forssel1999closed}. 
In \cite{ljung1996subspace}, it is 
demonstrated how subspace algorithms
may be applied to closed-loop data by 
fitting high order vector autoregressive models with exogenous inputs
(VARX models). The work of
\cite{jansson2003subspace} proposed a subspace
technique which used the VARX parameter estimates 
to recover the Kalman Filter.
This helped to develop algorithms such as 
the well known predictor based subspace identification (PBSID) algorithm
\cite{chiuso2007role}. 
%It is noted in much of this literature
%that most successful subspace approaches for closed-loop identification 
%have two steps. First, a high order VARX model is fit. Then, reduction steps are applied to obtain the 
%Kalman Filter, from which it is straightforward
%to determine the system model.
For summaries on the advancements of subspace approaches
for closed-loop identification, see \cite {qin2006overview} and \cite{van2013closed}. 

Our contribution is to analyze an algorithm for system identification
in which a VARX model is fit, followed by balanced model reduction. 
Such an approach has been described in \cite{dahlen2004relation}, 
and it was shown that its asymptotic properties match those of a familiar
subspace method known as canonical correlation analysis.
The primary difference of our analysis from prior
non-asymptotic system identification characterizations 
is that we allow the presence of a feedback controller. 

%We achieve bounds upon the expected one step prediction error, as well as high
%probability bounds on the $\Hinf$ norm of the error system from the estimated 
%model to the finite horizon Kalman Filter. %

The paper is organized as follows. In Section \ref{sec:prob}, we present the algorithm,
precisely define the problem, and provide the main result, a 
non-asymptotic bound on the generalization
error of the produced model. The proof of this result is available in Section \ref{sec:proof}. Section \ref{sec:discussion}
presents a related result regarding the high probability 
bounds on the $\Hinf$ norm of the error system from the identified model to the finite horizon
Kalman Filter, and highlights several practical considerations of the bounds.
The bound in expectation is then 
demonstrated on a randomly generated system in Section \ref{sec:simulation}.

%In addition to supplying a useful metric
%for performance analysis and comparison of closed-loop identification algorithms, 
%we hope that by demonstrating this
%route of analysis for closed loop systems, it will be possible to extend
%the analyses on controller design from \cite{tu2017non} and \cite{dean2017sample}
%to situations in which the identification is performed in closed-loop. This
%would allow an adaptive approach to system identification and robust control design. 

%%% Local Variables:
%%% mode: latex
%%% TeX-master: "gaussians"
%%% End:
%\input{prob}
\section{Problem and Results}
\label{sec:prob}
We now describe the details of the problem, and present
the generalization error bound obtained. Subsection \ref{ss:notation}
summarizes the notation used throughout the paper. In Subsection
\ref{ss:problem}, we highlight the details and assumptions of the
closed-loop system. The algorithm to be analyzed is presented
in \ref{ss:redar}, along with the main result, a non-asymptotic 
bound on the generalization error of the obtained model.  
\subsection{Notation and Terminology}
\label{ss:notation}
Random variables are denoted using bold symbols. 
The expected value of a
random variable, $\x$, is
denoted by $\E[\x]$, while the probability of an event $S$ is given by
$\P(S)$. 

The Euclidean norm of a vector, $x$, is denoted by $\| x\|$. The
Frobenius norm of a matrix, $G$, is denoted by $\|G\|_F$, while its
induced $2$ norm is denoted by $\|G\|$. The minimal eigenvalue of a
symmetric matrix, $X$, is denoted by $\lambda_{\min}(X)$. 

The \emph{power} of a stationary process, $\y_t$, is defined by
$\|\y\|_{\Pow}^2 = \E[\y_t^\top \y_t]$.

The forward shift operator is denoted by $q$, i.e. $q \x_t = \x_{t+1}$.
If $G(q)$ is a time-domain operator defined in terms of shifts, we will identify it with its
corresponding transfer matrix, $G(z)$. The $\Hinf$ norm of a transfer
matrix, $G(z)$, is denoted by $\| G \|_{\infty}$. 
The notation $\x_{i:j}$ represents the sequence starting from $\x_i$ and up to, but not including $\x_j$.

\subsection{Problem Setup}
\label{ss:problem}
    Consider a linear time-invariant (LTI) system in innovation form:
    \begin{subequations}
    \label{eq:innovations}
    \begin{align}
        \x_{t+1} &= A \x_t + B \u_t + K \e_t \\
        \y_t & = C \x_t + \e_t.
    \end{align}
    \end{subequations}
    Here $\x_t \in \R^{n_x}$ is the state, $\u_t \in \R^{n_u}$ is the known input, $\e_t \in \R^{n_y}$ is Gaussian white noise,
    and $\y_t \in \R^{n_y}$ is the measurement. For compact notation, we set $\z_t
    = \begin{bmatrix}
    \u_t^\top & 
    \y_t^\top
    \end{bmatrix}^\top \in \R^{n_z}$ 
    and $\d_t = \z_{t-p:t}$.
    For later analysis, we have assumed that the system is strictly proper
    in the known inputs, $\u_t$.

    We will assume that $\u_t$ can be represented as a linear feedback
    with excitatory noise:
    \begin{subequations}
    \label{eq:controlSystem}
    \begin{align}
    \s_{t+1} &= A^F\s_t + B_1^F \y_t + B_2^F \v_t \\
    \u_t &= C^F \s_t + D_1^F \y_t + D_2^F\v_t. 
    \end{align}
    \end{subequations}
    Here $\v_t \in \R^{n_u}$ is identity covariance Gaussian white noise
    which is independent of $\e_t$. $\s_t \in \R^{n_s}$ is the state of the
    controller. A summary of the system is shown in Fig. \ref{fig:overallLoop}.

    \begin{figure}
        \centering
        \begin{tikzpicture}[thick,>=latex]
        \node[draw] (P) {
            $\left[
            \begin{array}{c|cc}
                A & B & K \\
                \hline 
                0 & I & 0\\
                C & 0 & I
            \end{array}
            \right]$
            };

            \node[draw,below=3 em of P](C) {
            $\left[
                \begin{array}{c|cc}
                A^F & B_1^F & B_2^F \\
                \hline
                C^F & D_1^F & D_2^F
                \end{array}
            \right]$
            };

            \node[draw,left=2em of P](Eye) {$
            \begin{bmatrix}
                I & 0 \\
                0 & I
            \end{bmatrix}
            $};

            \node[draw,left=2em of Eye](D) {
            $
            \begin{bmatrix}
                z^{-1} I \\
                z^{-2} I \\
                \vdots \\
                z^{-p} I
            \end{bmatrix}
            $
            };

            \draw[->] (D.west) -- node[above] {$\d_t$} ($(D.west)+(-2em,0)$);
            \draw[->] (Eye.west) -- node[above] {$\z_t$} (D.east);
            \draw[->] let \p1=($(P.west)+(0,1em)$),\p2=(Eye.east) in (\p1)
            -- node[above] {$\u_t$} (\x2,\y1);
            \draw[->] let \p1=($(P.west)+(0,-1em)$),\p2=(Eye.east) in (\p1)
            -- node[above] {$\mathbf{y}_t$} (\x2,\y1);
            \draw[->] let
            \p1=($(P.west)+(0,-1em)$),\p2=(Eye.east),\p3=($(C.west)+(0,.5em)$),
            \p4=($(.5*\x1,\y1)+(.5*\x2,0)$),\p5=($(\x4,.5*\y4)+(0,.5*\y3)$),
            \p6=(\x5-2em,\y5),\p7=(\x6,\y3) in
            (\p4) -- (\p5) -- (\p6) -- (\p7) -- node[above] {$\mathbf{y}_t$} (\p3);
            \draw[->] let
            \p1=($(C.west)+(0,-1em)$),\p2=($(C.west)+(-2em,-1em)$) in
            (\p2) -- node[above]{$\v_t$}(\p1);

            \draw[->] let \p1=($(P.east)+(2em,1em)$),
            \p2=($(P.east)+(0,1em)$) in (\p1)-- node[above] {$\e_t$}
            (\p2);

            \draw[->] let \p1=(C.east), \p2=($(P.east)+(2em,-1em)$),
            \p3=(\x2,\y1),\p4=($(P.east)+(0,\y2)$) in
            (\p1) -- (\p3) -- (\p2)-- node[above] {$\u_t$}(\p4);
        \end{tikzpicture}
        \caption{\label{fig:overallLoop} The overall system. }
    \end{figure}
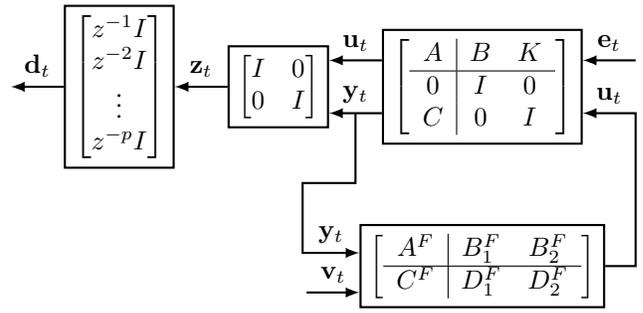

    The closed-loop system is assumed to be stable. This implies
    that the signal power, $\|\z\|_\Pow$, is finite. Additionally, we will
    assume that the joint covariance of the noise is positive definite:
    \begin{equation*}
    \E\left[
        \begin{bmatrix}
        \e_t \\
        D_2^F \v_t
    \end{bmatrix}
        \begin{bmatrix}
        \e_t \\
        D_2^F \v_t
    \end{bmatrix}^\top
    \right]
    =\begin{bmatrix}
    \Psi & 0 \\
    0 & \Omega
    \end{bmatrix} = \Gamma \succ 0.
    \end{equation*}
    This ensures that identifiability conditions hold, as in traditional
    system identification \cite{ljung1999system}. 
    Note that we do not assume that the open-loop system is stable. 

    The finite horizon Kalman Filter represents the output estimate 
    for \eqref{eq:innovations} provided the $p$ previous time steps as 
    \begin{align}
        \nonumber
        \yStar =  \E[\y_t| \d_t],
    \end{align}
    where the notation follows that mentioned previously;
the sequence $x_{i:j}$ does not include $x_j$. This indicates
that the finite horizon Kalman Filter estimate depends only upon data
collected at times $k$ with $t-p \leq k < t$.

    The estimate is a linear function of $\d_t$. We define $G_{OPT}$ as the transformation relating the two:
    \begin{align}
        \label{eq:filterMatrix}
        \yStar=G_{OPT} \d_t.
    \end{align}
    We also define the operator $H^{OPT}(q)$ such that
    \begin{align}
        \nonumber
        \yStar=H^{OPT}(q) \z_t.
    \end{align}
    The steady state Kalman Filter operator
    $H^\star(q)$ is given as  
    \begin{align}
        \nonumber
        \y^\star_{t|-\infty:t} = H^\star(q) \z_t.
    \end{align}
    Due to the special form of the innovations model, we can write the steady state Kalman Filter as 
    \begin{align*}
        \x_{t+1} &= (A-KC)\x_t + B \u_t + K \y_t \\
        \yKF &= C \x_t,
    \end{align*}
    and the associated expected squared error, $\E[ \| \y_t - \yKF\|^2]$, is $\|\e\|_\Pow^2$.

    \subsection{The REDAR Algorithm and its Prediction Error}
\label{ss:redar}
    \begin{algorithm}
    \begin{algorithmic}[1]
        \State Given signals $\u_{1-p:T+1}$, $\y_{1-p:T+1}$, VARX order $p>0$,
        a regularization paramater  $\alpha >0$, 
        and a reduction error $\phi >0$
        \State \label{line:d_def}
        Let $\z_t
    = \begin{bmatrix}
    \u_t^\top & 
    \y_t^\top
    \end{bmatrix}^\top$ and $\d_t = \z_{t-p:t}$
        \State Solve the VARX identification problem $\G_T = \argmin_{G}
        \sum_{t=1}^{T} \| \y_t - G \d_t \|^2 + \alpha \|G\|_F^2$
        \State Construct a state-space operator,  $\H^A$, such that $\H^A(q)
        \z_t = \G_T \d_t$
        \State Apply balanced model reduction to find $\H^R$ such that $\|
        \H^A- \H^R\|_{\infty} \le \phi$.
        \State Compute estimates $(\hat \A,\hat \B,\hat \C,\hat \D,\hat \K)$ by
        $\left[\begin{array}{c|cc}
                \hat \A - \hat \K \hat \C
                &  \hat \B & \hat \K \\
                \hline
                \hat \C & \hat \D & 0
            \end{array}
            \right]
            =                                      \H^R$
    \end{algorithmic}
    \caption{\label{alg:REDAR} The REDuced AutoRegressive (REDAR) algorithm }
    \end{algorithm}

    The method of this paper is termed the  REDAR (pronounced ``reader'')
    algorithm. See Alg.~\ref{alg:REDAR}. Here $\H^A$  is the system
    corresponding to the least-squares model, while $\H^R$ is the result of
    balanced model reduction subject to $\Hinf$ error tolerance $\phi$. See \cite{zhou1996robust} for a description balanced reduction with limited error tolerance.
    Our final predictor is given by $\hat
    \y_t = \H^R(q) \z_t$. Additionally,
    given the state-space realization of $\H^R$, all of the parameters of
    the innovation form model, \eqref{eq:innovations}, can be estimated. 

    The general scheme of the REDAR
    algorithm has been proposed in closed-loop system identification
    literature \cite{van2013closed}, \cite{dahlen2004relation}. However, its finite-sample
    behavior has not been characterized. Our main result gives such
    a characterization:

    \begin{theorem}
        \label{thm:mainResult}
        Suppose there exists $L > 0$ and $\rho < 1$ such that
        for all $|z| \geq \rho$, $\| H^\star(z)\| \leq L$.
        Then for all $T \geq T_0$,
        \begin{align}
        \nonumber
        \E[\| & \y_t-\hat{\y}_t \|^2 ] \\
        \nonumber
            & \leq \| \e\|_\Pow^2+ \frac{L \rho
        ^{p+1}}{1-\rho}\|\z\|_\Pow + 2 \phi \|\z\|_\Pow^2 
        + \frac{2kp}{\sqrt{T}} \|\z\|_\Pow^2,
        \end{align}
        where $k$ and $T_0$ are constants that depends upon $n_u$, $n_y$, $p$, $\alpha$, the $\Hinf$ norm of the closed
        loop system, and $\lambda_{min}(\Gamma)$.
    \end{theorem}

\section{Proof of Theorem \ref{thm:mainResult}}
The proof of Theorem~\ref{thm:mainResult} has several stages. 
In Subsection~\ref{ss:decomp}, the expected squared prediction error is
decomposed into terms due to 1) noise, 2) finite autoregressive order, 3) model
reduction, and 4) a limited amount of data. The error due to finite
autoregressive order is bounded in Subsection~\ref{subsec:finitemodel}. In
order to bound the errors due to limited data, some non-asymptotic
convergence results are derived in Subsection~\ref{subsec:convergence}. These
results are used to bound the error due to limited data in
Subsection~\ref{subsec:finitedata}. Finally, the errors due to model
reduction are bounded in Subsection~\ref{ss:reduction}. 

\label{sec:proof}

\subsection{Decomposition}
\label{ss:decomp}
The expected squared prediction error of Alg.
\ref{alg:REDAR} is now decomposed into the following components:
the optimal prediction error given the true model, two terms resulting from the limited model complexity
determined by the parameters $p$ and $\phi$, and a component dependent upon the limited amount of data. 
\begin{lemma}
    \label{lem:decomposition}
    let $\y_t^A = \G_T \d_t$ be the output of the VARX model.
    Then the prediction error of the REDAR algorithm can be decomposed as
    \begin{multline}
        \nonumber
        \E[\| \y_t-\hat{\y}_t\|^2] \leq \|\e\|_\Pow^2 +
        \E[\|\yKF - \yStar\|^2]\\
         + 2\E[\|\yStar - \y_t^A\|^2]+2\E[\|\y_t^A - \hat{\y}_t\|^2].
    \end{multline}
\end{lemma}

\begin{proof}
    \begin{align*}
    \MoveEqLeft
        \E[\| \y_t- \hat{\y}_t\|^2] = \E[\| \y_t-\yKF\|^2] \\
         +2\E  &  [  (\yKF-\hat{\y}_t)^\top (\y_t-\yKF)]
          +\E[\|\yKF-\hat{\y}_t\|^2].
    \end{align*}
    The first term on the right of the above expression is $\|\e\|_\Pow^2$.
    The second term can be seen to
    be zero by iterated expectation.
    The third term may be expanded as 
    \begin{align*}
    \MoveEqLeft
        \E[ \| \yKF-\hat{\y}_t\|^2]
        = \E[\|\yKF - \yStar\|^2] \\
        +2\E & [(\yKF-\yStar)^\top(\yStar - \hat{\y}_t)] \\ 
        + \E & [\| \yStar - \hat{\y}_t\|^2].
    \end{align*}
    
    Iterated expectation may be used to show that
    the second term on the right of the above
    equation evaluates to zero. 
   % \begin{align}
   %     \nonumber
   %     \E[(\yKF-\yStar)^\top (\yStar - \hat{\y}_t)] = 0.
   % \end{align}
   
   Now perform the following decomposition.   
   \begin{align}
       \nonumber
       \E[\| & \yStar -  \hat \y_t\|^2] = \E[\| \yStar - \y_t^A +
       \y_t^A- \hat \y_t \|^2] \\
       \nonumber
         &\le 2\left(\E[\| \yStar - \y^A_t\|^2] + \E[\| \y_t^A - \hat
         \y_t\|^2]
              \right), 
    \end{align}
    where the inequality follows from application of 
    the Cauchy Schwarz and triangle inequalities. 
\end{proof}

\subsection{Finite Model Order Error}
\label{subsec:finitemodel}
Here, we bound the term arising from Lemma
\ref{lem:decomposition} that results from the
finite model order:
\begin{align}
    \label{eq:finiteOrder}
    \E[\|\yKF - \yStar\|^2].
\end{align}

Recall that $H^\star(q)$ is the Kalman filter operator. 
Note that $H^\star(q)$ can be written as
\begin{align}
    \nonumber
    H^\star(q)=C \sum_{i=1}^\infty \tilde{A}^{i-1}[ B \quad K] q^{-i},
\end{align}
where $\tilde{A}=A-KC$. Let $H^{\TKF}$ be the truncation of
$H^\star(q)$ to $p$ terms:
\begin{align}
    \nonumber
    H^{\TKF}(q)=C \sum_{i=1}^{p} \tilde{A}^{i-1} [B \quad K] q^{-i}.
\end{align}
Then the difference between these two systems is
\begin{align}
  \nonumber
  H^{\Tail}(q)= &
        (H^{\star}-H^{\TKF})(q) \\
        \nonumber
        & = \sum_{i=p+1}^{\infty} C \tilde{A}^{i-1} [B
            \quad  K]q^{-i}.
\end{align}
To simplify notation, let $H_i=C \tilde{A}^{i-1} [B \quad K]$. 

Note that 
\begin{align}
    \nonumber
    \E[\| \yKF - \yStar\|^2 ] \leq \E[\| \yKF - \y_t^\TKF \|^2].
\end{align}
We therefore opt to bound the term on the right
hand side of the above equation. This may be 
written as
\begin{align}
    \nonumber
    \E[\|\yKF - \y_t^\TKF \|^2] = 
    \E[\|H^{\Tail}(q) \z_t\|^2].
\end{align}

For any operator $H$,
\begin{align*}
    \E[\|H(q) \x_t\|^2] \leq \|H\|_\infty^2 \| \x\|_\Pow^2.
\end{align*}

Thus we have that 
\begin{align}
    \label{eq:tailPower}
    \E[\|H^{\Tail}(q) \z_t\|^2] \leq
    \| H^\Tail \|_\infty^2 \|\z\|_\Pow^2.
\end{align}

\begin{lemma}
  \label{lem:kfTail}
(modification of \cite{goldenshluger2001nonasymptotic}, Lemma 1). Assume that there are constants $\rho <1$ and $L>0$ such that the
Kalman filter satisfies:
$\|H^\star(z)\|_2 \le L$ for all $|z| \ge \rho$. Then the coefficients of
$H^\star$ satisfy
\begin{equation}
  \nonumber
  \| H_i \|_2 \le L \rho^i \textrm{ for } i =1,2,\ldots
\end{equation}
and the tail is bounded as 
\begin{equation}
  \nonumber
  \|H^\Tail \|_\infty \le \frac{L\rho^{p+1}}{1-\rho}.
\end{equation}
\end{lemma}

\ifnum \value{num}=0
The proof of this lemma is omitted, as it is quite
similar to the proof provided in \cite{goldenshluger2001nonasymptotic}.
\else
\begin{proof}

By substituting $\frac{1}{z}$ into the expression for $H^\star(z)$, the following
is obtained. 
\begin{align}
    \nonumber
    H^\star(z^{-1})=\sum_{i=1}^{\infty} H_i z^i.
\end{align}

Let $\gamma$ be the counter-clockwise contour around $0$ of radius
$\rho^{-1}$.  Then for any constant, $a$, we have
\begin{equation}
  \nonumber
  \frac{1}{2\pi j}\oint_{\gamma} \frac{a}{z^k} dz=\begin{cases}
    a & \textrm{ if } k = 1 \\
    0 & \textrm{ if } k \ne 1
    \end{cases}
  \end{equation}
  Then the filter coefficients $H_i$ may be written as
\begin{align}
    \nonumber
    H_i=\frac{1}{2 \pi j} \oint_\gamma \frac{H^\star(\frac{1}{z})}{z^{i+1}} dz.
\end{align}
By setting $z=\rho^{-1}e^{j\theta}$,  the above
expression becomes 
\begin{align*}
    H_i=\frac{1}{2 \pi j} \int_0^{2\pi} \frac{H^\star(\frac{1}{\rho^{-1} e^{j
        \theta}})}{(\rho^{-1} \e^{j \theta})^{i+1}} j \rho^{-1} e^{j
        \theta} d \theta.
\end{align*}

Now the norm of $H_i$ may be bounded by
application of the triangle inequality and
homogeneity.
    \begin{align*}
            \|H_i\|_2 \leq \frac{1}{2 \pi} \int_0^{2 \pi} \frac{\|H^\star(\frac{1}{
                \rho^{-1} e^{j \theta}}) \|_2}{\rho^{-i}} d \theta.
    \end{align*}

The assumption implies that $\|H^\star(z^{-1})\|_2 \le L$ for all $|z|\le
\rho^{-1}$. Then $\|H^\star(\rho e^{-j\theta})\|_2 \leq L$ so $\|H_i\| \leq L \rho^i$. 
        \begin{align*}
            \|H^\star-H^{\TKF}\|_\infty \leq \sum_{i=p+1}^{\infty} L \rho^i = \frac{L
            \rho^{p+1}}{1-\rho}.
        \end{align*}
    \end{proof}
\fi
Combining the result of Lemma \ref{lem:kfTail} with \eqref{eq:tailPower}, we have a bound for \eqref{eq:finiteOrder}.
\subsection{Convergence of Empirical Means}
\label{subsec:convergence}
The least squares problem in Alg. \ref{alg:REDAR} 
converges asymptotically to some steady state value. 
This subsection takes the first step in bounding
the distance from the asymptotic value when a finite
amount of data is available. In particular,
probability bounds are provided for the difference
of individual components of the least squares solution
from their asymptotic value.

Recall the definition of $\d_t$  and the corresponding least-squares
estimator, $\G_T$, from Alg.~\ref{alg:REDAR}.
The least-squares solution can be expressed as 
\begin{align*}
  \Q_T &= \frac{1}{T} \sum_{t=1}^{T} \d_t\d_t^\top, \;
  \N_T = \frac{1}{T}
        \sum_{t=1}^{T} \y_t\d_t^\top,
        \\
  \G_T &= \N_T    
     \left (\Q_T+\frac{\alpha}{T}I \right)^{-1}.
\end{align*}

The optimal solution defined in \eqref{eq:filterMatrix} may 
be expressed as follows
\begin{align*}
    Q=\E[\Q_T], \;
    N=\E[\N_T], \;
    G_{OPT}=N Q^{-1}.
\end{align*}

We will denote $\Q_T-Q$ as $\Delta \Q$ and $\N_T-N$ as $\Delta \N$.
The focus of this subsection will be to derive a bound on the
probability that any element of $\Delta \Q$ or $\Delta \N$ exceed
a given magnitude. This will then be used in the following subsection
to bound the finite data error.

Let $J(q)$ be the closed-loop operator that maps 
\begin{align*}
\z_t = J(q) \begin{bmatrix}
    \Psi^{-1/2}\e_t \\
    \v_t
    \end{bmatrix}.
\end{align*}
Here, we have re-normalized the innovation error signal so that the
  input to $J$ is Gaussian white noise with identity covariance. 

Define
$
  \Z = \begin{bmatrix}
    \z_{1-p}^\top &
    \hdots &
    \z_{T-1}^\top &
    \z_{T}^\top
    \end{bmatrix}^\top.
$
%  Note that every entry of $\N_T$ and
%  $\Q_T$ can be expressed as linear functions
%  of $\Z\Z^\top$.
  Let $R = \E\left[\Z\Z^\top\right]$, and $r_{t} = \E[\z_t \z_{0}^\top]$ be the autocorrelation
    function. Then $R_{t,\tau} = r_{t-\tau}$.  Let $\Phi_{\z}(e^{j\omega})$
    be the Fourier transform of $r_t$, which is the
    power spectral density. Note that $\Phi_{\z}(e^{j\omega}) =
    J(e^{j\omega})J(e^{j\omega})^*$, and so
    $\|\Phi_{\z}(e^{j\omega})\| \le \| J\|_{\infty}^2$. 
 
  \begin{lemma}
    The covariance, $R$, satisfies $\|R\| \le \| J\|_{\infty}^2$.
  \end{lemma}
  \begin{proof}

    Let $v_t$ be a sequence such that $v_t=0$ for $t<1-p$ and
    $t>T$ and let $\hat v(e^{j\omega})$ be its Fourier transform. Let
    $v$ be the vector formed by stacking the components for
    $t=1-p,\ldots,T$.   Then 
  \begin{align*}
    v^\top R v &= \sum_{t,\tau=1-p}^{T} v_t^\top r_{t-\tau} v_{\tau} \\
    &= \frac{1}{2\pi} \int_0^{2\pi} \hat v(e^{j\omega})^*
      \Phi_{\z}(e^{j\omega}) \hat v(e^{j\omega}) d\omega \\
    &\le \|J\|_{\infty}^2 v^\top v .
  \end{align*}
  The second equality follows from Plancharel's theorem, while the
  inequality uses the bound on $\|\Phi_{\z}(e^{j\omega})\|_2$, followed by
  Plancharel's theorem again. The lemma now follows by maximizing over
  unit vectors, $v$. 
\end{proof}

\begin{lemma}
  \label{lem:matrixConcentration}
  For all symmetric $S$ and all $\delta >0$, the following bound holds
  for all $T \ge p$.
  \begin{multline*}
    \P\left(\Z^\top S\Z > \Tr(RS) + \delta T\right) \le \\
    \exp\left(
      -T \min\left\{
        \frac{\delta^2}{32\|S\|^2 \|J\|_{\infty}^4},\frac{\delta}{8
          \|S\| \|J\|_{\infty}^2}
        \right\}
      \right).
  \end{multline*}
\end{lemma}

\begin{proof}
Note that $SR$ and $R^{1/2}SR^{1/2}$ have the same eigenvalues, so
  all of the eigenvalues of $SR$ are real. For all $\eta >0$ such that
  $\eta R^{1/2}SR^{1/2} \prec I$, 
  Markov's inequality implies that 
  \begin{align}
    \nonumber
    \MoveEqLeft
    \P\left( 
    \Z^\top S \Z > \Tr(SR) +T \delta
    \right) \\
    \nonumber
    & \le e^{-\frac{\eta}{2} (\Tr(SR)+T\delta)} \E\left[
              e^{\frac{\eta}{2} \Z^\top S \Z}
              \right] \\
              & = e^{\left( -\frac{1}{2} \left(
              \eta \Tr(SR)  + \eta T\delta + \log\det(I-\eta SR)
              \right) \right)}.
              \label{eq:chiSquaredExponent}
  \end{align}
  The equality follows from direct calculation.

  Now we will examine the exponent from
  \eqref{eq:chiSquaredExponent}. 
  Let $\lambda_i$ be the eigenvalues of $SR$ for $i=1,\ldots,T+p$. As discussed above,
  these are real and furthermore,
  \begin{equation*}
  |\lambda_i| \le \|S R\| \le \|S\| \|R \| \le \| S\| \|J\|_{\infty}^2.
\end{equation*}
Using the bounds on the eigenvalues, the exponent can be bounded as
follows. 
\begin{align*}
  \MoveEqLeft
  \eta \Tr(SR) + \eta T\delta + \log\det(I- \eta SR)
  \\
  &=
 \sum_{i=1}^{T+p}\left( \eta \lambda_i + \log(1-\eta \lambda_i)\right)
      + \eta T \delta \\
  &= \eta T \delta -\sum_{i=1}^{T+p} \sum_{k=2}^{\infty} \frac{(\eta
    \lambda_i)^k}{k} \\
  & \ge \eta T \delta - 2T \sum_{k=2}^{\infty}\frac{(\eta \|S\|
    \|J\|_{\infty}^2)^k}{k} \\
  & \ge \eta T \delta - 2T \sum_{k=2}^{\infty} (\eta \|S\|
    \|J\|_{\infty}^2)^k \\
  &= \eta T \delta - 2T \frac{(\eta
    \|S\| \|J\|_{\infty}^2)^2}{1-\eta \|S\| \|J\|_{\infty}^2}.
\end{align*}

Now say that $\eta \le 1/(2\| S\|_2 \|J\|_\infty^2)$. Then
the above expression can be bounded  below by
\begin{equation}
  \label{eq:removedDenominator}
  T \left(
    \eta \delta - 4 \eta^2 (\| S\|_2 \|J\|^2_{\infty})^2
  \right).
\end{equation}
Now we will see how to choose $\eta$ to ensure that
\eqref{eq:removedDenominator} is positive. For simple notation, let $a
= 8 (\| S\|_2 \|J\|_{\infty}^2)^2$ and let $b = 1/(2\| S\|_2
\|J\|_\infty^2)$. Then $\eta$ can be chosen by solving the following
maximization problem:
\begin{subequations}
  \begin{align*}
    & \max && \eta \delta - \frac{a\eta^2 }{2} \\
    & \textrm{subject to} && 0 \le \eta \le b
  \end{align*}
\end{subequations}
The optimal solution is given by $\eta = \min\{\delta/a,b\}$. If $\eta
= \delta / a$, then the optimal value is given by $\delta^2/(2a)$. If
$\eta = b$, then we must have that $\delta \ge ab$ and so the optimal
value satisfies
\begin{equation*}
  b\delta - ab^2/2 \ge b\delta - b\delta /2 = b\delta / 2.
\end{equation*}
Thus, we get the final bound on \eqref{eq:removedDenominator} as
\begin{equation*}
  T \min\left\{
    \frac{\delta^2}{16\|S\|^2 \|J\|_{\infty}^4},\frac{\delta}{4
      \|S\| \|J\|_{\infty}^2}
    \right\}.
\end{equation*}
The lemma follows by plugging this into the exponential bound on the
probability from \eqref{eq:chiSquaredExponent}.
\end{proof}

   Note that every entry of $\N_T$ and
   $\Q_T$ is of the form
   \begin{equation*}
    \frac{1}{T}\sum_{t=1}^{T} (\z_{t-k})_i (\z_{t-\ell})_j
  \end{equation*}
   for some $i,j \in \{1,\ldots,n_z\}$ and $k,\ell \in
   \{0,\ldots,p\}$. Recall that $r_t$ is the autocorrelation function
   of $\z$. The following lemma shows that these empirical means
   converge to the corresponding autocorrelation values exponentially
   in probability. 
   
    \begin{lemma}
      \label{lem:elemConcentration}
      For all $i,j \in \{1,\ldots,n_z\}$, all $k,\ell \in
      \{0,\ldots,p\}$, and all $T\ge p$, the following bound holds
      \begin{multline*}
         \P\left(
         \left|\frac{1}{T} \sum_{t=1}^{T} (\z_{t-k})_i (\z_{t-\ell})_i -
        (r_{\ell - k})_{ij}\right| > \delta
          \right)  \\
          \le 2 \exp\left(
         -T \min\left\{
           \frac{
              \delta^2}
             {32 \| J\|_{\infty}^4},\frac{\delta}{8 \|J\|_{\infty}^2}
         \right\}
      \right).
      \end{multline*}
    \end{lemma}

    \begin{proof}
      We will apply Lemma~\ref{lem:matrixConcentration}. To do so,
      we will express the sums as quadratic forms. Note that
  \begin{equation*}
    \sum_{t=1}^{T} (\z_{t-k})_i (\z_{t-\ell})_j = \Z^\top S \Z = \frac{1}{2}\Z^\top(M+M^\top)\Z,
  \end{equation*}
  where
  \begin{equation*}
    M = \begin{bmatrix}
      0_{(p-k)\times (p-\ell)} & 0_{(p-k)\times T} & 0_{(p-k)\times
        \ell} \\
      0_{T\times (p-\ell)} & I_T & 0_{T\times \ell} \\
      0_{(k\times (p-\ell)} & 0_{k\times T} & 0_{k\times \ell}
    \end{bmatrix} \otimes (e_i e_j^\top).
  \end{equation*}
  Here $e_i,e_j \in\R^{n_z}$ are the canonical unit vectors and the
  subscripts in the matrix on the left denote the dimensions. 

  Note that there are permutation matrices, $P_L$ and $P_R$ such that
  \begin{equation*}
    P_L M P_R = \begin{bmatrix}
      I_T & 0 \\
      0 & 0
      \end{bmatrix}
    \end{equation*}
    for zero matrices of appropriate size.
    Thus $\|M\|_2 = 1$ and so $\|S\| \le 1$, by the triangle
    inequality and homogeneity.

    Furthermore, in this case
    \begin{equation*}
      \Tr(SR) = \E[\Z^\top S\Z] = T(r_{\ell-k})_{ij}.
      \end{equation*}
    
    Since the bound from Lemma~\ref{lem:matrixConcentration} increases
    with respect to $\|S\|$, we can plug in the upper bound of $1$
    to show that
    \begin{multline*}
      \P\left(
          \frac{1}{T} \sum_{t=1}^{T} (\z_{t-k})_i (\z_{t-\ell})_i -
         (r_{\ell - k})_{ij} > \delta
           \right)  \\
           \le \exp\left(
          -T \min\left\{
            \frac{
              \delta^2}
              {32 \| J\|_{\infty}^4},\frac{\delta}{8 \|J\|_{\infty}^2}
         \right\}
      \right).
    \end{multline*}
    
    The probability bound on $-\frac{1}{T} \sum_{t=1}^{T} (\z_{t-k})_i (\z_{t-\ell})_i +
         (r_{\ell - k})_{ij}$ is identical, and follows by
    applying Lemma~\ref{lem:matrixConcentration} to $-S$. The lemma
    now follows from a
    union bound. 
  \end{proof}

    Now note that every element of $\Delta \Q$ and $\Delta \N$ may
    be expressed as 
   \begin{equation*}
    \frac{1}{T}\sum_{t=1}^{T} (\z_{t-k})_i (\z_{t-\ell})_j-(r_{l-k})_{ij}
  \end{equation*}
   for some $i,j \in \{1,\ldots,n_z\}$ and $k,\ell \in
   \{0,\ldots,p\}$. The following lemma uses this fact
   to bound the probability of elementwise deviations of $\Delta \Q$
   and $\Delta \N$ from zero.
  
\begin{lemma}
    \label{lem:allElements}
    For $T \geq p$, $\delta \geq 0$, the probability that any element
    of $\Delta \N$ or $\Delta \Q$ is larger than $\delta$ in magnitude
    satisfies 
    \begin{multline*}
        \P( \underset{i, j}{\max} \{ |\Delta \Q_{ij}| \} > \delta
            \text{ or }  \underset{i, j}{\max} \{ | \Delta \N_{ij} | \} > \delta) \\
        < 2b  \exp \left( -T \min \left \{ \frac{\delta^2}{32
        \|J\|_\infty^4}, \frac{\delta}{8 \|J\|_\infty^2} \right
        \} \right) 
    \end{multline*}
    where
    \begin{align*}
        b= p n_y n_z +\frac{p n_z (p n_z +1)}{2}.
    \end{align*}
\end{lemma}
\ifnum \value{num}=0
The proof of this lemma is omitted. It follows from a union
bound and application of Lemma \ref{lem:elemConcentration}.
\else
\begin{proof}
    By a union bound, 
    \begin{multline*}
        \P( \underset{i, j}{\text{max}} \{ |\Delta \Q_{ij}| \} > \delta
            \text{ or } \underset{i, j}{\text{max}} \{ | \Delta \N_{ij} | \} > \delta) \\
        \leq (\text{dim}(\Delta \Q)+\text{dim}(\Delta \N)) \P(|\Delta \Q_{11}| > \delta). 
    \end{multline*}
    Where an arbitrary element in $\Delta \Q$ and $\Delta \N$ is represented by
    with $\Delta \Q_{11}$. Noting that $\Delta \Q$ is symmetric, we assign
    \begin{align*}
        b = & \text{dim}(\Delta \N)+\text{dim}(\Delta \Q)\\
           & = p n_y n_z +\frac{p n_z (p n_z +1)}{2}.
    \end{align*}
    The lemma now follows by applying Lemma \ref{lem:elemConcentration} to bound $\P( | \Delta \Q_{11} | > \delta)$.
\end{proof}
\fi

\subsection{Finite Data Error}
\label{subsec:finitedata}
We now use the results from the previous subsection to determine a bound 
for
\begin{align*}
    \E[\|\yStar-\y^A\|^2] \leq
    \E[\| G_{OPT} - \G_T\|^2] \|\z\|_\Pow^2 p.
\end{align*}
As $\|G_{OPT} - \G_T\|^2 \geq 0$ the expected value may be written
\begin{align}
    \nonumber
    \E[\| & \G_T- G_{OPT} \|^2] \\
    = &\int_0^\infty \P [\|\G_T-G_{OPT}\|^2
    > \epsilon] d \epsilon .
    \label{eq:cdfIntegral}
\end{align}
An upper bound on this integral may be computed if, for any $\epsilon \geq 0$, we can bound
 $\P[\| \G_T-G_{OPT} \|^2 > \epsilon]$. To do so, 
define $\delta \geq 0$ such that
\begin{align*}
    |\Delta \N_{ij}| \leq & \delta \; i=1, \dots, n_y \; j=1, \dots, p n_z\\ 
    |\Delta \Q_{ij}| \leq & \delta \; i,j =1, \dots, p n_z .
\end{align*}
We will proceed by bounding $\|\G_T-G_{OPT}\|$ in terms of $\delta$. It will
then be possible to determine a value $\delta \geq 0$ corresponding to all sufficiently 
large $\epsilon$ such that
\begin{align*}
    \nonumber
    |\Delta  \Q_{ij}| \leq \delta \text{ and } |\Delta \N_{ij}| \leq \delta
    & \implies \\
    & \|\G_T-G_{OPT}\|^2 \leq \epsilon.
\end{align*}
Then Lemma \ref{lem:allElements} may be applied to bound the probability that the elementwise
bounds hold. 

The elementwise bounds above provide the following bounds on 
$\| \Delta \N\|$ and $\| \Delta \Q\|$.
\begin{subequations}
    \label{eq:deltaBounds}
    \begin{align}
        \|\Delta \N\| \leq & c_1 \delta \\
        \|\Delta \Q\| \leq & c_2 \delta 
    \end{align}
\end{subequations}
where $c_1=\sqrt{p n_y n_z}$ and $c_2=p n_z$.

To simplify notation in the following analysis, we define
 $\xi=\lambda_{min}(\Gamma) \geq \lambda_{min}(Q) = \|Q^{-1}\|^{-1}$.

\begin{lemma}
    \label{lem:normBound}
    \begin{align*}
        \nonumber
         \|  \G_T -G_{OPT} \| 
         \leq \left(c_3 \delta+ \frac{c_4}{T}\right) \left \|(Q+\Delta \Q+\frac{\alpha}{T}I)^{-1} \right \|
    \end{align*}
    where $c_3 = c_1 + \frac{\|J\|_\infty^2 c_2}{\xi}$
    and $c_4 = \frac{\|J\|_\infty^2 \alpha}{\xi}$.
\end{lemma}
\begin{proof}
    \begin{align*}
        \nonumber
        \G_T & -G_{OPT} 
        = (N + \Delta \N)(Q+\Delta \Q+\frac{\alpha}{T} I)^{-1}-N Q^{-1}
    \end{align*}
    Application of the matrix inversion lemma provides
    \begin{align*}
        \nonumber
        & \G_{T} -G_{OPT} \\
        & = (\Delta \N-N Q^{-1}(\Delta \Q+ \frac{\alpha}{T} I))(Q+\Delta \Q+\frac{\alpha}{T} I)^{-1}.
    \end{align*}
    If we now take the two norm, and apply both the triangle inequality
    and submultiplicativity several times, we get
    \begin{align*}
        \nonumber
        \|& \G_T  -G_{OPT} \| \leq \\
        \nonumber
        & \left(\frac{\|N\| (\|\Delta \Q\|+\frac{\alpha}{T})}{\xi} + \| \Delta \N\| \right)
         \left \|(Q+\Delta \Q+\frac{\alpha}{T}I)^{-1} \right \|.
    \end{align*}
    To bound $\|N\|$ in terms of $\|J\|_\infty$, note that we can write $N$ in terms of $R$ as 
    \begin{align*}
        N = \begin{bmatrix}
        0_{n_y \times p n_z + n_u} & I_{n_y} & 0_{n_y \times (T-2)n_z}
        \end{bmatrix} 
        R 
        \begin{bmatrix}
        I_{p n_z} \\
        0_{p n_z \times (T-1)n_z}
        \end{bmatrix}
    \end{align*}
so 
\begin{align*}
    \|N\| \leq \|R\| \leq \|J\|_\infty^2.
\end{align*}
    The lemma now follows from \eqref{eq:deltaBounds}.
\end{proof}

We know that the following always holds
    \begin{align}
        \label{eq:alldelta}
        \|(Q+\Delta \Q+\frac{\alpha}{T}I)^{-1}\| \leq \frac{T}{\alpha}, 
    \end{align}
as $Q+\Delta \Q= \sum_{k=0}^T \d_t \d_t^T \succeq 0$. A tighter bound is available when $\delta$ is small.

\begin{lemma}
    \label{lem:smalldelta}
    For $\delta < \frac{(
    \xi-\frac{\alpha}{T})}{c_2}$, 
    \begin{align*}
        \|(Q+\Delta \Q+\frac{\alpha}{T} I)^{-1}\|  \leq \frac{1}{\xi - c_2 \delta -\frac{\alpha}{T}}
    \end{align*}
\end{lemma}
\begin{proof}
    \begin{align}
        \label{eq:denom}
        (Q& +\Delta \Q+\frac{\alpha}{T}I)^{-1} 
        =Q^{-1} (I + \Delta \Q Q^{-1}+ \frac{\alpha}{T} Q^{-1})^{-1}.
    \end{align}
    When 
    \begin{align}
       \label{eq:eigvalBound}
        \|\Delta \Q Q^{-1}+ \frac{\alpha}{T} Q^{-1}\|
        < 1,
    \end{align}
    the term on the right may be replaced by its series expansion: 
    \begin{align*}
        (I  + \Delta \Q Q^{-1}+ \frac{\alpha}{T} Q^{-1})^{-1}
         = \sum_{k=0}^\infty (-\Delta \Q Q^{-1} - \frac{\alpha}{T} Q^{-1})^k.
    \end{align*}
    It is now possible to bound the two norm of \eqref{eq:denom} with 
    submultiplicativity and the triangle inequality.
    \begin{align*}
        \|(Q +\Delta \Q+\frac{\alpha}{T}I)^{-1} \| 
        \leq \frac{1}{\xi - \| \Delta \Q \| -\frac{\alpha}{T}}.
    \end{align*} 
    The condition in \eqref{eq:eigvalBound} can be seen to hold for $\delta < \frac{(
    \xi-\frac{\alpha}{T})}{c_2}$ by noting that
    \begin{align*}
         \| \Delta \Q Q^{-1}+ \frac{\alpha}{T} Q^{-1}\| \leq\frac{c_2 \delta+\frac{\alpha}{T}}{\xi}.
    \end{align*}
\end{proof}
\begin{lemma}
    \label{lem:finddelta}
    Assume $T \geq \max \left \{ \frac{2 \alpha}{\xi}, 1 \right \} $. Let 
    \begin{align*}
        \epsilon_0 =\left(\frac{2 \|J\|_\infty^2 \alpha }{\xi^2 T^{1/4}} \right)^2
    \text{ and }
        \epsilon_1 = \left(\frac{2 \|J\|_\infty^2 T}{\alpha} \right)^2.
    \end{align*}
    
    For any $\epsilon \geq \epsilon_0$, we
    can find $\delta \geq 0$ such that
\begin{align*}
    \nonumber
    |\Delta  \Q_{ij}| \leq \delta \text{ and } |\Delta \N_{ij}| \leq \delta
    & \implies \\
    & \|\G_T-G_{OPT}\|^2 \leq \epsilon,
\end{align*}
    by selecting
\begin{subequations}
\label{eq:deltas}
\begin{align}[left = {\delta = \empheqlbrace}]
    \label{eq:delta1}
    & \frac{(\xi T - \alpha) \sqrt{\epsilon} - c_4}{c_2 T \sqrt{\epsilon} + c_3 T} & \epsilon_0 \leq \epsilon \leq \epsilon_1 \\
    \label{eq:delta2}
    & \frac{\alpha \sqrt{\epsilon} - c_4}{c_3 T}  & \epsilon \geq \epsilon_1
\end{align}    
\end{subequations}
\end{lemma}
\begin{proof}
The conditions on T along with the definition
of $\epsilon_0$, $\epsilon_1$, and $c_4$ guarantee that \eqref{eq:deltas} is greater than or equal to zero. It can
be seen that expression \eqref{eq:delta1}
is less than $\frac{\xi - \frac{\alpha}{T}}{c_2}$ for
all values of $\epsilon$,
thus the condition in Lemma \ref{lem:smalldelta} 
is satisfied by \eqref{eq:delta1}.
The lemma follows by plugging \eqref{eq:delta1} into
the bound on $\|\G_T-G_{OPT}\|$ resulting from
Lemmas \ref{lem:normBound} and \ref{lem:smalldelta},
and \eqref{eq:delta2} into the bound on
$\|\G_T-G_{OPT}\|$ resulting from
Lemma \ref{lem:normBound} and \eqref{eq:alldelta}. 
\end{proof}

The reason for
the two different expressions for $\delta$ in the
above lemma is that \eqref{eq:alldelta} provides a
tighter bound than Lemma \ref{lem:smalldelta} when
$\delta$ becomes greater than $\frac{\xi - \frac{2 \alpha}{T}}{c_2}$.
Leveraging this advantage is of crucial importance in the following lemma. 

\begin{lemma}
    \label{lem:ARXerror}
    For some $k$ and $T_0$ depending on $n_u$, $n_y$, $p$, $\alpha$, $\|J\|_\infty$, and $\lambda_{min}(\Gamma)$, 
    \begin{align*}
         \E[\|\G_T-G_{OPT}\|^2] \leq \frac{k}{\sqrt{T}}.
    \end{align*}
    for all $T \geq T_0$.
\end{lemma}
\begin{proof}
Let $\delta_1(\epsilon)$ be given by \eqref{eq:delta1}
and $\delta_2(\epsilon)$ be given by \eqref{eq:delta2}.
We obtain a bound on the right side of
\eqref{eq:cdfIntegral} by application of Lemma
\ref{lem:finddelta} along with Lemma \ref{lem:allElements}.
\begin{align*}
    & \int_0^\infty (\P[\| \G_T-G_{OPT} \|^2 > \epsilon]) d
    \epsilon
    \leq \underbrace{\int_{0}^{\epsilon_0} 1 d \epsilon}_{d_1} \\
    \nonumber
    & +\underbrace{\int_{\epsilon_0}^{\epsilon_1} 
    2b \exp\left(-\frac{T}{2} \text{min }
    \left \{\left(\frac{\delta_1(\epsilon)}{4
    \|J\|_\infty^2}\right)^2,
    \frac{\delta_1(\epsilon)}{4 \|J\|_\infty^2} \right
    \}\right)  d \epsilon}_{d_2} \\
    \nonumber
    & + \underbrace{\int_{\epsilon_1}^{\infty} 
    2b \exp \left(-\frac{T}{2} \text{min }
    \left \{\left(\frac{\delta_2(\epsilon)}{4
    \|J\|_\infty^2}\right)^2,
    \frac{\delta_2(\epsilon)}{4 \|J\|_\infty^2} \right 
    \}\right)  d \epsilon}_{d_3},
\end{align*}
where the integrand of $d_1$ results from the fact that the probability is at most $1$. The bounds used above were valid for T $\geq \max \left \{\frac{2 \alpha}{\xi}, p \right \}$.
We now bound each term $d_1, d_2$ and $d_3$ separately.
\subsubsection{Bound $d_1$}
$d_1$ evaluates to to $\epsilon_0 = \frac{k_1}{\sqrt{T}}$ where
\begin{align*}
    k_1 = \frac{4 \|J\|_\infty^4 \alpha^2}{\xi^2}.
\end{align*}

\subsubsection{Bound $d_2$} 
Assign
\begin{align*}
    \gamma_1(\epsilon) = \frac{\delta_1(\epsilon)}{4 \|J\|_\infty^2},  \; 
    \gamma_2(\epsilon) = \frac{\delta_2(\epsilon)}{4 \|J\|_\infty^2}.
\end{align*}

Note that $\gamma_1(\epsilon)$ is monotonically increasing
for $\epsilon \geq 0$. This can be seen by observing that 
$\frac{d \gamma_1(\epsilon)}{d \epsilon} \geq 0 \text{ for } \epsilon \geq 0$. As a result, we have that for any $c \geq 0$, $\gamma_1(\epsilon) \geq
\gamma_1(c)$ for all $\epsilon \geq c$. Then if we define a constant $c_5$
such that $\epsilon_0 \leq c_5 \leq \epsilon_1$, we obtain the following
result.
\begin{align}
\nonumber
    & \int_{\epsilon_0}^{\epsilon_1} 
    2b \exp\left(-\frac{T}{2} \text{min }
    \left \{\gamma_1(\epsilon),
    \gamma_1(\epsilon)^2 \right
    \}\right)  d \epsilon \\
    \nonumber
    & = \int_{\epsilon_0}^{c_5} 
    2b \exp\left(-\frac{T}{2} \text{min }
    \left \{\gamma_1(\epsilon),
    \gamma_1(\epsilon)^2 \right
    \}\right)  d \epsilon \\
    \nonumber
    & + \int_{c_5}^{\epsilon_1} 
    2b \exp\left(-\frac{T}{2} \text{min }
    \left \{\gamma_1(\epsilon),
    \gamma_1(\epsilon)^2 \right
    \}\right) d \epsilon\\
   \nonumber
    & \leq 2b c_5 \left( 
    \exp \left(-\frac{T \gamma_1(\epsilon_0)}{2} \right) + \exp \left(-\frac{T \gamma_1(\epsilon_0)^2}{2}  \right) \right) \\
    & + 2b \epsilon_1 \left( 
    \exp \left(-\frac{T \gamma_1(c_5)}{2} \right) + \exp \left(-\frac{T \gamma_1(c_5)^2}{2}  \right) \right).
    \label{eq:middleUpper}
\end{align}

Recall that 
\begin{align*}
    \gamma_1(\epsilon_0) = \frac{(\xi T - \alpha) \sqrt{\epsilon_0} - c_4}{ 4\|J\|_\infty^2 (c_2  \sqrt{\epsilon_0} +  c_3 )T}. 
\end{align*}  
The condition $T \geq \frac{2 \alpha}{\xi}$ tells us that $\xi T -\alpha \geq \frac{\xi T}{2}$. Then, plugging in the expressions for $\epsilon_0$ and $c_4$, the numerator is bounded below by
\begin{align*}
    \frac{\|J\|_\infty^2 \alpha}{\xi} (T^{3/4} -1).
\end{align*}
A crude bound on this may be obtained by noting that for $T \geq 4, \; (T^{3/4}-1) \geq
 \frac{T^{3/4}}{2}$.

The denominator of $\gamma_1(\epsilon_0)$ can be bounded
above as
\begin{align*}
    4 \|J\|_\infty^2 \left(c_2 \frac{2 \|J\|_\infty^2 \alpha}{\xi^2 T_0^{1/4}}+c_3 \right) T
\end{align*}
by noting that $\epsilon_0$ decreases with increasing $T$. Thus if we plug in some $T_0 \geq 1$, an upper bound on the denominator is obtained for $T \geq T_0$. The resultant bound on $\gamma_1(\epsilon_0)$ is
\begin{align*}
    \gamma_1(\epsilon_0) \geq \frac{c_6}{T^{1/4}},
\end{align*}
where
\begin{align*}
    c_6 = \frac{\alpha}{8 \xi \left(2 c_2  \frac{\|J\|_\infty^2 \alpha}{ \xi^2 T_0^{1/4}}+c_3 \right)}.
\end{align*}
We may also bound $\gamma_1(c_5)$ by writing
\begin{align*}
    \gamma_1(c_5) = \frac{(\xi T - \alpha) \sqrt{c_5} - c_4}{ 4\|J\|_\infty^2 (c_2  \sqrt{c_5} +  c_3 )T} \\
    \geq \frac{(\frac{\xi T}{2}) \sqrt{c_5} - c_4}{ 4\|J\|_\infty^2 (c_2  \sqrt{c_5} +  c_3 )T}. 
\end{align*}
Then, if we set $c_5 = \left( \frac{4 c_4}{\xi} \right)^2$, we obtain a bound for the expression above as 
\begin{align*}
    \frac{c_4 }{ 4\|J\|_\infty^2 (c_2  \frac{4 c_4 }{\xi} +  c_3 )} = c_7
\end{align*}
for $T \geq 1$. 
Note that for the value of $c_5$ to lie between $\epsilon_0$ and $\epsilon_1$, we must have T satisfy the following condition.
\begin{align*}
    T \geq \max \left \{
    1,
    \frac{2 \alpha^2}{\xi^2} \right \}.
\end{align*}
Now \eqref{eq:middleUpper} may be bounded above by
\begin{align}
    \nonumber
    \label{eq:midInt}
    & c_8 \exp \left(-\frac{c_6 T^{3/4}}{2}\right) + c_8 \exp(-\frac{c_6^2 \sqrt{T}}{2}) \\
    & + c_9 T^2 \exp \left(-\frac{c_7}{2} T \right) + c_9 T^2 \exp \left(-\frac{c_7^2}{2} T \right)
\end{align}
where
\begin{align*}
    & c_8 = 2 b c_5, \; c_9 = \frac{8 b \|J\|_\infty^4}{\alpha^2}.
\end{align*}

\subsubsection{Bound $d_3$}
$d_3$ may be bounded as
\begin{align*}
    \nonumber
    \int_{\epsilon_1}^{\infty} &
    2b \exp \left(-\frac{T}{2} \text{min }
    \left \{\gamma_2(\epsilon)^2,
    \gamma_2(\epsilon)^2 \right 
    \}\right)  d \epsilon \\
    \nonumber
    & \leq \underbrace{\int_{\epsilon_1}^{\infty} 
    2b \exp \left(-\frac{T}{2} \gamma_2(\epsilon)^2\right) d \epsilon}_{d_4} \\
    & + \underbrace{\int_{\epsilon_1}^{\infty}
    2b \exp \left(- \frac{T}{2} \gamma_2(\epsilon) \right) d \epsilon}_{d_5}.
\end{align*}

We can evaluate $d_4$ as
\begin{align}
    \label{eq:exponentialTail}
    \nonumber 
    \int_{\epsilon_1}^{\infty} &
    2b \exp \left(- \frac{T}{2} \gamma_2(\epsilon) \right) d \epsilon \\
    & = 4b \lambda(\sqrt{\epsilon_1}+\lambda) \exp \left(\frac{\|J\|_\infty^2}{\xi \lambda} \right)
    \exp \left(-\frac{\sqrt{\epsilon_1}}{\lambda} \right),
\end{align}
with
\begin{align*}
    \lambda = \frac{8 \|J\|_\infty^2 c_3}{\alpha}.
\end{align*}

Plugging $\epsilon_1$ to the right hand side 
of \eqref{eq:exponentialTail} provides 
\begin{align}
    \nonumber
    \label{eq:expInt}
    4b \lambda \left(\frac{2\|J\|_\infty^2 T}{\alpha}+\lambda \right) \exp \left(\frac{\|J\|_\infty^2}{\lambda} \left(
\frac{1}{\xi}-\frac{2T}{\alpha} \right) \right) \\
%    \exp \left(-\frac{2\|J\|_\infty^2 T}{\alpha \lambda} \right) \\
    = c_{10} T \exp \left(-\frac{T}{c_{12}} \right)+ c_{11} \exp \left(-\frac{T}{c_{12}} \right),
\end{align}
with 
\begin{align*}
    c_{10} &= 8b \lambda \frac{ \|J\|_\infty^2}{\alpha} \exp \left(\frac{\|J\|_\infty^2}{\xi \lambda} \right), \\
    c_{11} &:= 4b \lambda^2 \exp \left(\frac{\|J\|_\infty^2}{\xi \lambda} \right) , \;
    c_{12} = \frac{\alpha \lambda}{2 \|J\|_\infty^2}.
\end{align*}

Meanwhile, $d_5$ may be
bounded by applying the Gaussian tail bound
$\frac{1}{\sqrt{2\pi}}\int_{z}^{\infty}e^{-x^2/2}
dx \le \sqrt{\frac{2}{\pi}}\frac{1}{z}e^{-z^2/2}$
. See exercise 2.11 in \cite{wainwright2019high}.
\begin{align}
\nonumber
\int_{\epsilon_1}^{\infty} &
    2b \exp \left(-\frac{T}{2} \gamma_2(\epsilon)^2\right) d\epsilon\\
    & \leq \left(4b \sigma^2 T +8b\sigma T \frac{\|J\|_\infty^2}{y_1 \xi} \right) \exp \left (-\frac{y_1^2}{2 T} \right),
     \label{eq:gaussianTail}
\end{align}
where 
\begin{align*}
    \sigma = \frac{4 c_3 \| J\|_\infty^2}{\alpha}, \;
    y_1 = \frac{\sqrt{\epsilon_1} - \frac{\|J\|_\infty^2}{\xi}}{\sigma}.
\end{align*}

Substituting $\epsilon_1$ into our definition for $y_1$ provides
\begin{align*}
    y_1 = \frac{\frac{2 \|J\|_\infty^2 T}{\alpha} - \frac{\|J\|_\infty^2}{\xi}}{\sigma}.
\end{align*}
$T \geq \frac{2 \alpha}{\xi}$ gives the loose bound
\begin{align*}
    y_1 \geq \frac{\|J\|_\infty^2 T}{\sigma \alpha}.
\end{align*}
Then the right side of \eqref{eq:gaussianTail}
is bounded below by
\begin{align}
    \label{eq:gaussianInt}
    c_{13} T \exp \left( \frac{-T}{c_{15}} \right)
+ c_{14}  \exp \left( \frac{-T}{c_{15}} \right),
\end{align}
where 
\begin{align*}
    c_{13} = 4 b \sigma^2, \;
    c_{14} = \frac{8 b \alpha \sigma^2}{\xi}, \;
    c_{15} = 2 \frac{\sigma \alpha}{\|J\|_\infty^2}.
\end{align*}

Note that each of the eight terms of \eqref{eq:midInt}, \eqref{eq:expInt}, and
\eqref{eq:gaussianInt} may be expressed in the form
\begin{align*}
    a_i T^{\left(m_i-1/2\right)} \exp(-b_i T^{n_i}).
\end{align*}
To complete our proof, we must find constants $k_i$ for $2 \leq i \leq 9$
 such that each of these terms is bounded as
\begin{align}
    \label{eq:boundPolyExp}
    a_i T^{\left(m_i-1/2 \right)} \exp(-b_i T^{n_i}) \leq \frac{k_i}{\sqrt{T}}
\end{align}
for $T \geq T_0$. To do so, note that
\begin{align*}
    a_i T^{m_i} \exp(-b_i T^{n_i})
\end{align*}
is maximized by $T_{\max, i} = \left(\frac{m_i}{n_i b_i} \right)^{1/n_i}$,
and monotonically decreasing for $T \geq T_{\max, i}$. Then if we choose
$T_0 \geq T_{max, i}$, and set 
\begin{align*}
    k_i = a_i T_0^{m_i} \exp(-b_i T_0^{n_i}),
\end{align*}
\eqref{eq:boundPolyExp} is satisfied.

Thus we set $k = \sum_{i=1}^9 k_i$ and 
require $T_0$ to be greater than or
equal to
$\max \left \{\frac{2 \alpha}{\xi}, p, 4,
\frac{2 \alpha^2}{\xi^2}, T_{\max, i} \right \}$.
\end{proof}

\subsection{Model Reduction Error}
\label{ss:reduction}
The only term that remains to be bounded is the one arising from 
the model reduction step. 
The bound on this term arises from the fact that
\begin{align*}
    \E[\|\y_t^A - \hat{\y}\|^2] \leq \|\H^A - \H^R\|_\infty^2 \|\z\|_\Pow^2.
\end{align*}
The balanced reduction step of Alg. \ref{alg:REDAR} guarantees that $\|\H^A -
\H^R\|_\infty \leq \phi$. 

Theorem \ref{thm:mainResult} now follows by applying Lemma \ref{lem:decomposition} to split
the expected squared error of our estimate into
the optimal estimator squared error, the finite model order error, the model reduction error, and the finite date error. Subsection
\ref{subsec:finitemodel} demonstrates the bound on the finite model order error. The finite data error is bounded in Subsection \ref{subsec:finitedata}. 

\section{Discussion}
\label{sec:discussion}
A slightly different result following from the same
analysis is provided below, along with a couple of notes regarding the
error bounds obtained.
\begin{theorem}
    \label{thm:mainresult2}
    for $0 < \theta \leq 1$, let
    \begin{align*}\delta = 4 \|J\|_\infty^2 \max \left\{\frac{2}{T}\log{\frac{2b}{\theta}}, \sqrt{ \frac{2}{T} 
    \log{\frac{2b}{\theta}}}\right\}.
    \end{align*} 
    Assume $T \geq p$. With probability at least $1-\theta$,
    \begin{align*}[left = {\|H^{OPT} - \H^R\|_\infty \leq \empheqlbrace \; \;}]
        & \frac{\left(c_3 \delta+ \frac{c_4}{T}\right)}{\xi - c_2 \delta -\frac{\alpha}{T}} p +
        \phi & \delta \leq \frac{\xi - \frac{2\alpha}{T}}{c_2}\\
        &  \frac{T \left(c_3 \delta+ \frac{c_4}{T}\right)}{\alpha} p +
        \phi & \delta > \frac{\xi - \frac{2\alpha}{T}}{c_2}
    \end{align*}
\end{theorem}
\begin{proof}
    By the triangle inequality,
    \begin{align*}
        \|H^{OPT} - \H^R\|_\infty \leq 
        \|H^{OPT} - \H^{A}\|_\infty +
        \|\H^{A} - \H^{R}\|_\infty.
    \end{align*}
    The second term is limited to be at most
    $\phi$ in Alg. \ref{alg:REDAR}. The first term
    may be bounded as 
    \begin{align*}
        \|H^{OPT} - \H^{A}\|_\infty \leq
        p \|G_{OPT} - \G_T\|.
    \end{align*}
    Applying Lemma \ref{lem:allElements} to $\delta$
    defined above yields that 
    \begin{align*}
    \MoveEqLeft
        \P( \underset{i, j}{\max} \{ |\Delta \Q_{ij}| \} < \delta
            \text{ and }  \underset{i, j}{\max} \{ | \Delta \N_{ij} | \} < \delta)
        \geq 1-\theta.
    \end{align*}
    The theorem now follows by bounding $\|G_{OPT} - \G_T\|$ with Lemmas
    \ref{lem:normBound} and \ref{lem:smalldelta}
    for $\delta \leq \frac{\xi - \frac{2\alpha}{T}}{c_2}$, and
    Lemma \ref{lem:normBound} and \eqref{eq:alldelta} for $\delta > \frac{\xi - \frac{2\alpha}{T}}{c_2}$.
\end{proof}

\begin{remark}
There are multiple free parameters left in the bound from Theorem \ref{thm:mainResult}. In particular, $\rho$ may be chosen as
any value between the spectral radius of the kalman filter and one. 
A smaller value of $\rho$ will increase L, but decrease $\frac{\rho^{p+1}}{1-\rho}$. As such, we can optimize over $\rho$ numerically to obtain 
the tightest bound. $T_0$ is also a free parameter, able to take any
value greater than that supplied in Lemma \ref{lem:ARXerror}. Choosing higher values
of $T_0$ will decrease the value of $k$, at the cost of making the bound
invalid for small values of $T$.
\end{remark}  

\begin{remark}
In practical application, the engineer does
not have access to all of the variables that are used to compute the bound a
priori. It is,
however, possible to estimate these from data. For instance, one could
perform an iterative approach in which a model with high complexity is used to 
obtain a rough estimate for system parameters before fitting 
a model with lower complexity. Similar ideas
are described in \cite{goldenshluger2001nonasymptotic}, 
\cite{goldenshluger1998nonparametric},
 and \cite{sarkar2019finite}.
\end{remark}

\section{Simulation}
\label{sec:simulation}
To test the derived error bound, random plants and controllers were created such
that the closed loop system was stable. The plants had the form of \eqref{eq:innovations},
while Linear Quadratic Gaussian controllers with random weight matrices and added noise, having the
the form of \eqref{eq:controlSystem} were used.
Algorithm~\ref{alg:REDAR} was applied to data generated from the closed loop
systems. It was seen that for each VARX model order $p$, and truncation bound
$\phi$, the prediction error was below the error bound at all timesteps. This result is shown below for one system with multiple values of $p$ and $\phi$. The bound is shown in orange, while the prediction error on a set of test
data as a function of $T$ in Alg.~\ref{alg:REDAR} is shown in blue.
\begin{figure}[H]
\includegraphics[scale=0.6]{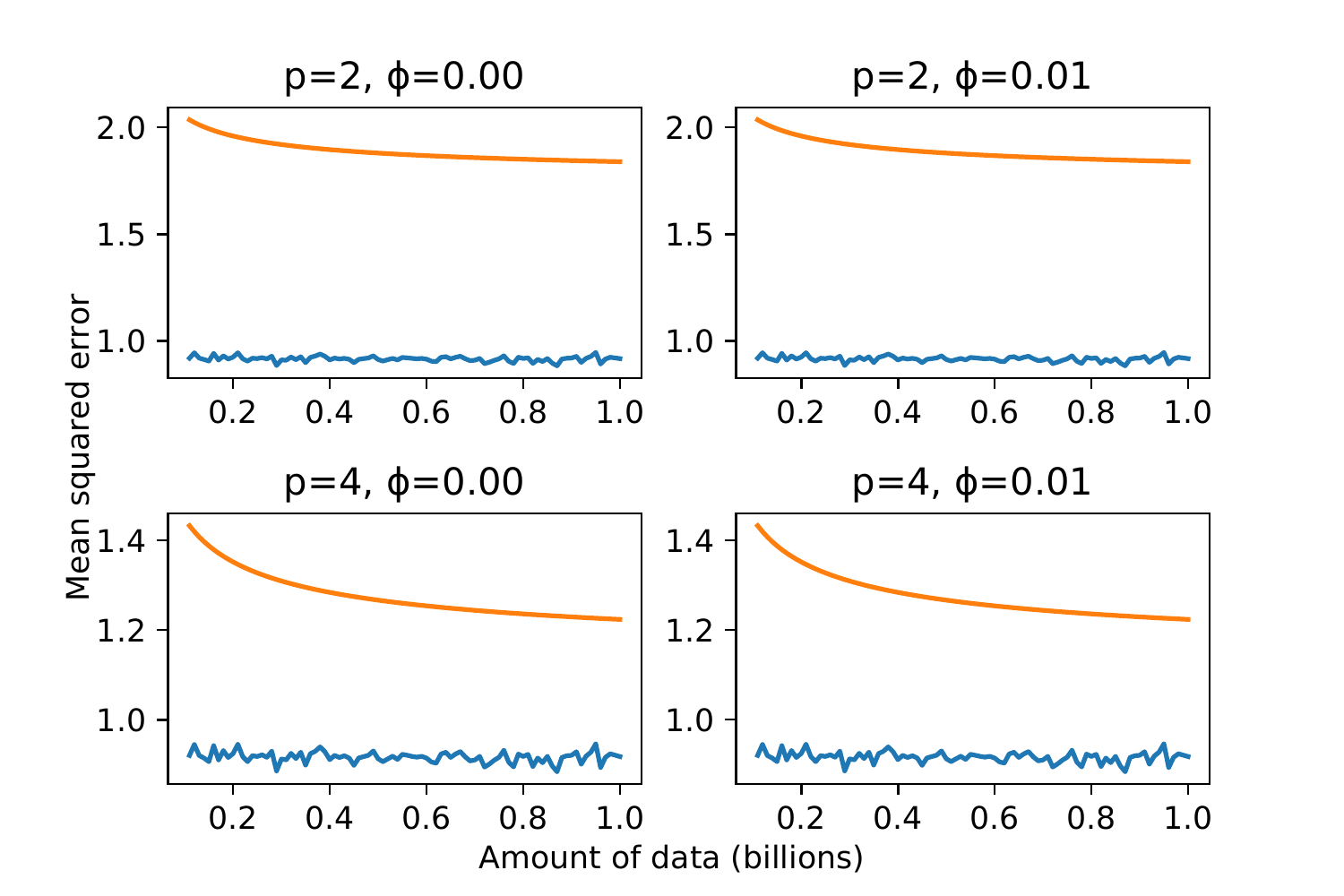}
\caption{Prediction errors of the REDAR algorithm run on a randomly generated system are seen to 
fall below the error bound at all times.}
\end{figure}
It should be noted that the bound is clearly not 
tight on the system above. A tighter bound could be obtained by removing
several of the cruder bounding techniques, and
choosing free parameters in the bound more carefully.
This effort was not undertaken in this this work. 

\section{Conclusion}
The finite sample behavior of
an algorithm known as REDAR was characterized for data generated in closed-loop.
The algorithm follows an approach used by many identification 
methods in which the data is fit to a VARX model, and the
system model is obtained via a reduction step. 
Due to the simple nature of the algorithm, it 
was possible to derive a non-asymptotic upper bound on the generalization error. Though the bound is not tight, it provides
the engineer with a notion of the effectiveness of
the model with a finite amount of data, which
allows for comparison of algorithms and parameter
selection for the model. Additionally, high probability bounds
on the $\Hinf$ norm of the error system from the estimated
model to the finite horizon Kalman Filter are obtained. It may be possible to utilize these 
bounds for robust control synthesis. 
As the analysis holds for identification of 
closed loop systems, this would allow
for an adaptive approach to robust control design to 
be applied. 

\section{Acknowledgements}
The authors thank Jianjun Yuan for helpful discussions
regarding the finite data bound.

\bibliography{ref}

\end{document}